\documentclass[twocolumn,prl,aps,showpacs]{revtex4-1}
\newcommand{\be}{\begin{equation}}
\newcommand{\ee}{\end{equation}}
\usepackage{graphicx}
\usepackage{color,amsthm}
\theoremstyle{plain}
\bibstyle{prsty}
\begin{document}
\title{Operational Interpretation of the Ratio of Total Correlations to Classical Correlations}
\author{Sai Vinjanampathy$^1$ and A. R. P. Rau$^2$}

\affiliation{ 
$^1$Centre for Quantum Technologies, National University of Singapore, 3 Science Drive 2, Singapore 117543.\\
$^2$Department of Physics and Astronomy, Louisiana State University, Baton Rouge, LA 70803, USA.
}

\begin{abstract}
We will discuss the generalization of entropic uncertainty principles in terms of a game. The game involves $k$-players, each measuring one of $k$ possible observables. The question is, what is the maximum number of players that can  play such that their joint entropic uncertainties are tightened by the presence of non-classical correlations? We answer this question and relate it to the ratio of quantum mutual information to classical correlations. This game hence serves to construct an operational interpretation of the aforementioned ratio. It provides another example of quantum correlations providing a quantum advantage.
\end{abstract}
\date{\today}
\maketitle
\section{Introduction}
Quantifying the difference between the classical and the quantum world has been a major goal of physics \cite{bell1966problem,haroche1998entanglement}. There have been many definitions of ``quantumness", that do not necessarily agree with each other. Various indicators of quantumness, namely violation of Bell's inequality, entanglement \cite{horodecki2009quantum}, quantum discord \cite{henderson2001classical,ollivier2001quantum}, etc., do not necessarily agree with each other. For instance, two-qubit Werner states \cite{werner1989quantum}, whose parameter $p$ that indicates the admixture of uniform noise, are classical from the standpoint of accommodating a
hidden variable theory when $p$ lies below $1/\sqrt{2}$. However, such states are considered quantum with respect to the presence of entanglement for $1/3 < p < 1/\sqrt{2}$. Furthermore, for values $0<p<1/3$, they have no entanglement but have non-zero discord \cite{ali2010quantum}, another measure of non-classicality. These discrepancies often simply point to the fact that different questions are being asked and answered with respect to what constitutes the ``quantum" world. 

Another such measure of quantumness has been the complementarity of observables in quantum physics. While in classical physics, one can measure both position and momentum simultaneously with arbitrary precision for such pairs of complementary observables, there is a bound on the accuracy of such measurements in quantum physics, which can be understood as a measure of non-classicality. On the other hand, contextuality reveals no such non-classicality for a single qubit \cite{grudka2008there,peres1995quantum}.

 We will present an operational definition of the ratio of the total correlations to the classical correlations in terms of an entropic uncertainty principle (EUP) \cite{Wehner}. In doing so, we will present a way to understand non-classicality in terms of these two measures. We note that the difference of total correlations to classical correlations is quantum discord \cite{henderson2001classical,ollivier2001quantum}, another measure of non-classicality.

We will begin by briefly introducing entropic uncertainty principles in Section 2. After introducing previous results on entropic uncertainty principles in the presence of quantum memory, we will present the  notion of ``generalized $k$-party uncertainty games" in Section 3. We will then present a new entropic uncertainty relation for $k$-party games and use it to prove the operational definition of the ratio between the total and classical correlations in Section 4. We present conclusions in Section 5.

\section{Entropic Uncertainties}
Non-commutativity is one of the hallmarks of quantum theory that sets it apart from classical mechanics. The Heisenberg \cite{Heisenberg} and Robertson \cite{Robertson} inequalities highlight that non-commuting observables cannot be simultaneously measured with arbitrary precision. The product of the uncertainties of any two simultaneously measured observables of a quantum system A is bounded by 
\be\label{rob}
\Delta R_{1}\Delta R_{2}\geq \frac{1}{2}\big|\langle [R_1,R_2]\rangle\big|,
\ee
where, $\Delta R_1:=\sqrt{\langle R_1^2\rangle-\langle R_1\rangle^2}$ is the standard deviation of the measurement and $[R_1,R_2]:=R_1R_2-R_2R_1$ is the commutator. 

Many authors found two aspects of the uncertainty principle in Eq.(\ref{rob}) unsatisfying: Deutsch \cite{Deutsch} argued that the standard deviation was a bad measure for multi-modal distributions and hence must be replaced with an entropy measure (see also \cite{Hirschmann, Beckner}) since it naturally quantifies the statistical variability of measurement outcomes. Deutsch also pointed out a second issue, that the right-hand side of Eq.(\ref{rob}) is state dependent. This was seen to be undesirable since the ``incompatibility" between any two operators should not depend on the states in which they are being measured. It should, Deutsch argues, be a property of the Hilbert space the operators are defined on. Otherwise, for a pair of observables in finite dimension, since there always exists a state which reduces the Robertson bound to zero \cite{Wehner}, that inequality becomes trivial. 

Deutsch addressed these issues by providing an entropic uncertainty relation that was bounded by a number that was only a function of the chosen observables, and not the state. This relation was improved upon by a conjecture by Kraus \cite{Kraus}, which was proven by Maassen and Uffnik \cite{MaassenUffnik}. They showed that the largest uncertainty was obtained when the operators were maximally non-commuting among themselves. Such measurements are called mutually unbiased measurements \cite{bandyopadhyay2002new}. This entropic uncertainty relation can be stated as
\be
H(R_1)+H(R_2)\geq-\log{c(R_1,R_2)}.
\ee
Here $H(R_{1,2})$ is the Shannon entropy of the probability distribution generated by the measurement $R_{1,2}$ and $c(R_1,R_2):=\max_{i,j}\vert\langle r^{(1)}_{i}\vert r^{(2)}_{j}\rangle\vert^2$. $\vert r^{(1,2)}_{i}\rangle$ are eigenvectors of $R_{1,2}$. Hence $c(R_1,R_2)$ represents the maximum overlap between any two eigenvectors with one drawn from each set.

In \cite{BertaRenner}, Berta \textit{et.al.,} extended this entropic uncertainty relation to include entanglement. To discuss this extension, they introduced the notion of ``uncertainty games." A two-person uncertainty game involves many copies of a bipartite state $\rho_{AB}$. Alice gets one of the subsystems, say A, from each copy. She chooses one of two measurement operators $R_1$ and $R_2$ and performs repeated measurements with A. Now, given that Bob has subsystem B, which may be entangled to A, he enjoys access to more information than was available in the absence of the memory subsystem A. Berta \textit{et.al.,} showed that this information serves to lower the uncertainty bound when Bob measures the other operator that Alice did not choose. Their improved entropic uncertainty relation is given by
\be
S(R_1\big| B)+S(R_2\big| B) \geq-\log{c(R_1,R_2)}+S(A\big|B),
\ee
where $S(R_1\big| B)$ is the conditional von Neumann entropy associated with the measurement $R_1$, given the presence of the memory B. The additional term, $S(A\big| B)$, represents the conditional entropy of the state $\rho_{AB}$. When this conditional entropy is negative in the presence of entanglement, the last term lowers the bound.

This result was further extended by \cite{PatiUshaWilde} who noted that the presence of quantum correlations beyond a certain value serves to tighten the bound. This entropic relation is given by
\begin{widetext}
\be
S(R_1\big| B)+S(R_2\big| B) \geq-\log{c(R_1,R_2)}+S(A\big|B)+\mathrm{\max}{\{0,D_{A}-\mathcal{C}\}}.
\ee
\end{widetext}
Here $D_{A}$ is the one-way quantum discord associated with measurements on A and $\mathcal{C}$ is the maximum conditional information associated with the post-measurement state, and is referred to as  the classical correlation associated with the measurements on system A. Note that in \cite{PatiUshaWilde}, $\mathcal{C}$ is referred to as $J_{A}$. Note also that while $S(A\big|B)$ could be negative and hence decrease the uncertainty bound, the last term, $\mathrm{\max}{\{0,D_{A}-\mathcal{C}\}}$ serves to always tighten the bound.

\section{$k$-Party uncertainty game}
We wish to generalize this uncertainty relation and use it to get an operational interpretation of the ratio of quantum to classical correlations. We begin by recollecting some important definitions. First, the quantum mutual information is given by
\be
\mathcal{I}=S(A)-S(A\big| B).
\ee
The maximum conditional entropy associated with the post measurement state, namely the classical correlations, are given by
\be
\mathcal{C}=S(B)-\min\displaystyle\sum_{i}p_{i}S(\rho_{i}).
\ee

Here $p_i$ are the probabilities of obtaining a measurement result corresponding to effect operator $E_i=(K^{\dagger}_i\otimes\mathbf{I})(K_i\otimes\mathbf{I})$ and $\rho_i$ is the corresponding state $\rho_i=(K_i\otimes\mathbf{I})\rho_{AB}(K^{\dagger}_i\otimes\mathbf{I})/p_i$. Now, this leads to the equation
\be
\mathcal{I}-k\mathcal{C}=S(A)-S(A\big| B)-kS(B)+k\min\displaystyle\sum_{i}p_{i}S(\rho_{i}).
\ee
or, equivalently,
\be\label{minusk}
-kS(B)=\mathcal{I}-k\mathcal{C}-S(A)+S(A\big| B)-k\min\displaystyle\sum_{i}p_{i}S(\rho_{i}).
\ee

We are now in a position to play a $k$-party game. The game is stated as follows: Assume Bob has access to many copies of a bipartite state $\rho_{AB}$, while $(k-1)$-players named Alice$_1$ through Alice$_{k-1}$ access the physical system A. (Alternatively, a single Alice may measure $(k-1)$ different operators.) Each Alice picks from a list of $k$ operators, labelled $R_{1}$ through $R_{k}$, one operator not previously chosen, and measures it. This results in say Alice$_1$ measuring $R_5$, Alice$_2$ measuring $R_{3}$ (because $R_5$ was already chosen), and so on. When all $(k-1)$ Alices have picked their unique operators from the list of $k$ operators, Bob is left with one final operator. Now, given that Alice$_{1}$'s measurement yielded an entropic uncertainty of $S(R_5\big| B)$, Alice$_{2}$'s a similar uncertainty of $S(R_{3}\big|B)$, and so on, \textit{what is the entropic uncertainty bound on the operator Bob measured}?
\begin{center}
\begin{figure}
\begin{center}  
\includegraphics[height=65mm,width=65mm]{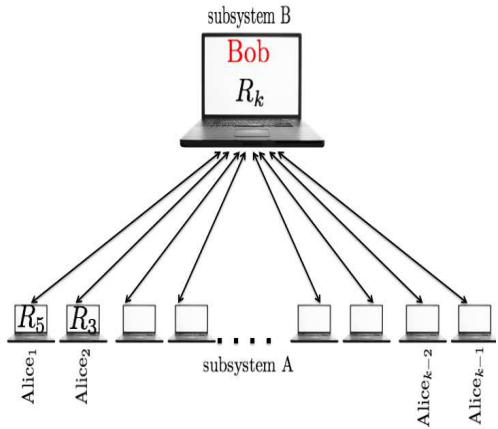}  
\caption{Bob has system B of a bipartite state $\rho_{AB}$. Each Alice possesses a copy of the sub-system A. Each of the $(k-1)$ Alices must pick from a list of $k$ measurements operators, one operator she chooses to measure, leaving Bob with one final measurement operator. The figure illustrates Alice$_1$ picking the measurement operator $R_5$, Alice$_2$ choosing $R_3$, and so on, resulting in Bob being left with the measurement operator $R_k$.}  
\label{IoC}
\end{center}  
\end{figure}
\end{center} 

To answer this question, we start by writing the following statement about the conditional entropies associated with each measurement, namely
\be 
S(R_{i}\big| B)=H(R_{i})+\displaystyle\sum_{j}p^{(i)}_{j}S(\rho^{(i)}_{j})-S(B).
\ee
Here $p^{(i)}_{j}$ measures the probability of the $j$th outcome when measurement $R_i$ is performed and likewise $\rho^{(i)}_{j}$ is the associated post-measurement density matrix. Adding all $k$ terms up, we get
\be
\displaystyle\sum_{i=1}^{k}S(R_{i}\big| B)=\displaystyle\sum_{i=1}^{k}\left(H(R_{i})+\displaystyle\sum_{j}p^{(i)}_{j}S(\rho^{(i)}_{j})\right)-kS(B).
\ee
Substituting for the last term $-kS(B)$ from Eq.(\ref{minusk}), we get
\begin{widetext}
\be \label{central}
\displaystyle\sum_{i=1}^{k}S(R_{i}\big| B)=\left\{\displaystyle\sum_{i=1}^{k}H(R_{i})-S(A)\right\}+S(A\big| B)+\Bigg(\mathcal{I}-k\mathcal{C}\Bigg)+\left(\displaystyle\displaystyle\sum_{i=1}^{k}\sum_{j}p^{(i)}_{j}S(\rho^{(i)}_{j})-k\min\displaystyle\sum_{l}p_{l}S(\rho_{l})\right).
\ee
\end{widetext}

The right-hand side of the equation is composed of four terms, grouped under different parentheses. The first term is the sum of the Shannon entropies associated with the probability distribution generated by each measurement minus $S(A)$ and has to be optimized over the set of all states. Before discussing this in detail, we will briefly discuss the other terms. The second term is the conditional entropy introduced by Berta \textit{et.al.}, whose negativity heralds a quantum advantage. The third term is the difference between the mutual information and $k$ times the classical correlations. The last term is a positive constant as seen by noting that since the minimization is over the set of all POVMs, the quantity $\min\sum_{l}p_{l}S(\rho_{l})$ is at least as big as the smallest quantity in the sum $\sum_{i=1}^{k}\sum_{j}p^{(i)}_{j}S(\rho^{(i)}_{j})$, which represents the post-measurement entropies associated with $k$-measurements. 

Now, returning to the first term in Eq.(\ref{central}), when minimized over all states, for $k=2$, and for projective measurements, it yields $-\log(c)$, the result by Maassen and Uffnik \cite{MaassenUffnik}. Its generalization to $k$ generalized measurements can be very hard to obtain since it involves a complex optimization problem. The answer in specific cases involving mutually unbiassed bases is well known \cite{Wehner}. Since the sum $\sum_{i=1}^{k}H(R_{i})-S(A)$ must scale with the number of measurements, we can in general write $\sum_{i=1}^{k}H(R_{i})-S(A)=-k\log{(c')}$, where $c'$ absorbs all of the details of the measurements. For $k=2$, $2\log{c'}=\log{c}$. Note that $-k\log(c')+S(A\big| B)\geq0$, since it has already been proven for the case of $k=2$. If we drop the last term, Eq.(\ref{central}) can be replaced with an inequality.

This makes the $k$-party EUP to read
\begin{widetext}
\be\label{k-part}
\displaystyle\sum_{i=1}^{k}S(R_{i}\big| B) \geq -k\log(c')+S(A\big| B)+\mathcal{I}-k\mathcal{C}.
\ee
\end{widetext}
Note that when $\mathcal{I}-k\mathcal{C}<0$, which happens for increasing $k$, for a given value of $\mathcal{I}$ and $\mathcal{C}$, the inequality becomes increasingly uninteresting. This is due to the fact that for any such inequality, the addition of an increasingly negative term on the right hand side of the inequality trivially makes the inequality stronger. This completes the proof.

\section{Operational Interpretation of $\mathcal{I}/\mathcal{C}$}
Eq.(\ref{k-part}) brings us to the following theorem:
\newtheorem{kk}{Theorem}
\begin{kk}
In a $k$-party uncertainty game, the maximum number of players that can play before the advantage from non-classical correlations in tightening the bound is lost is given by $k_{opt}=\lfloor\mathcal{I}/\mathcal{C}\rfloor$.
\end{kk}
\begin{proof}
Since $\mathcal{I}-k\mathcal{C}\leq0$ would mean that the corresponding term will be zero, the solution to the equation $\mathcal{I}-k\mathcal{C}=0$ sets the limit to the maximum number of players who can play the game and extract an advantage from quantum correlations. This solution is trivially given by $k_{opt}$. Note that since $D_{A}>0\Leftrightarrow \mathcal{C}>0$ \cite{modi2011quantum}, and since the advantage is only available if $D_{A}>\mathcal{C}$, the limit of zero classical correlations is not relevant.
\end{proof}

We thus have an operational interpretation of the ratio of quantum to classical correlations. Note that although $\mathcal{I}$ is bounded by $\mathcal{I}\leq2\min(\log(d_A),\log(d_B))$, $\mathcal{C}$'s lower bound is simply zero. This lower bound is not relevant as explained in the proof. In fact, although the tightening of the bound by non-classical correlations implies the presence of quantum discord, the presence of quantum discord is insufficient to herald a tightening of the inequality.

Finally, we note an interesting limit. In the limit of vanishing correlations, since $J_A$ may vanish faster than $D_A$, an infinite number of players can play the game.
\begin{center}
\begin{figure}[h!]  
\begin{center}  
\includegraphics[height=43.3mm,width=70mm]{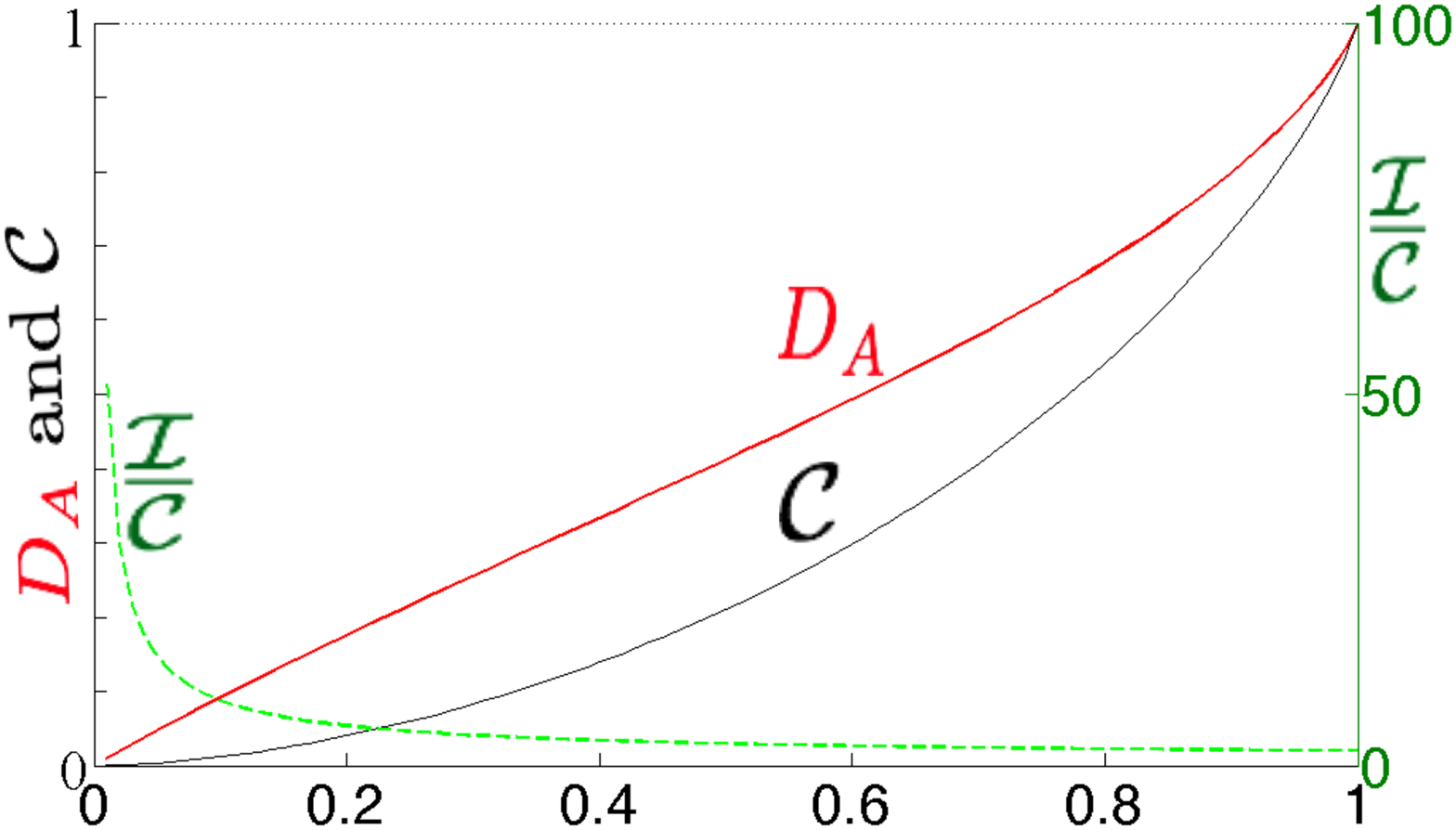}  
\caption{Quantum discord($D_A$), classical correlations ($J_A$) and the ratio of total correlations ($\mathcal{I}$) to $J_A$ plotted for the state $\rho=a\vert\psi^+\rangle\langle\psi^+\vert+(1-a)\vert11\rangle\langle11\vert$ for various values of the parameter $a$. Note that at $a=0.01$, $\lfloor\mathcal{I}/J_A\rfloor=52$.}  
\label{IoC}
\end{center}  
\end{figure}
\end{center} 

An example is provided by the state \cite{ali2010quantum} 
\be
\rho=a\vert\psi^{+}\rangle\langle\psi^{+}\vert+(1-a)\vert11\rangle\langle11\vert,
\ee
where $\vert\psi^{+}\rangle=(\vert01\rangle+\vert10\rangle)/\sqrt{2}$. In Fig. (\ref{IoC}), we plot the ratio of total to classical correlations. The ratio of total to classical correlations is rather large for small values of the parameter $a$. Indeed, for a=0.01, the ratio of total to classical correlations is $\approx$52.02. This means that a maximum of 52 players can play the uncertainty game. Furthermore, if we choose the three standard mutually unbiassed bases namely $(\vert0\rangle,\vert1\rangle)$, $(\vert\pm x\rangle)$, $(\vert\pm y\rangle)$, then the lower bound $min(\sum_i H(R_i))$ was shown \cite{sanches1998optimal} to be $2$. $S(A)=0.0454$, $S(A|B)$ is 0.0354 and hence $\mathcal{I}=0.01$.

\section{conclusions}
In this work, we presented an operational interpretation of the ratio of total correlations to classical correlations. Such interpretations have often been pursued to gain a better understanding of non-classical measures and correlations \cite{dodd2002simple,boixo2008operational,cavalcanti2011operational,madhok2011interpreting}. We introduced the notion of $k$-party uncertainty games wherein the objective of the game was to minimize Bob's measurement uncertainty, given that $(k-1)$ Alices had performed individual measurements on their subsystem. Our analysis shows how this ratio can be interpreted in terms of Bob minimizing the total uncertainty of post measurement entropies of $k$ joint measurements. To minimize the total post measurement entropy, while tightening the uncertainty bound, Bob can include at most $k_{opt}=\lfloor\mathcal{I}/\mathcal{C}\rfloor$ number of Alices. We note that such a tightening of the EUP will find applications in problems such as tightening of the security of cryptographic protocols and memory assisted measurements of multiple observables.

\begin{acknowledgments}
ARPR thanks the Alexander von Humboldt Foundation for support through a Wiedereinladung and Dr. Gernot Alber of Technische Universit\"{a}t, Darmstadt, for hospitality while this work was completed. This work was partially completed when SV was at the University of Massachusetts, Boston during which time he was supported by NSF under Projects No. PHY -0902906 and PHY-1005571. SV acknowledges discussions with Kavan Modi. The Centre for Quantum Technologies is supported by National Research Foundation and Ministry of Education, Singapore.
\end{acknowledgments}
\bibliography{EUP_As_Games_v2}
\end{document}